\documentclass[11pt]{article}
\usepackage{amssymb}\usepackage{amsmath}
\usepackage{amsthm}
\newcommand{\tr}{\mathop{\rm tr}\nolimits}
\def\ST{{\mathsf T}}
\def\SP{{\mathsf P}}

\def\CT{{\mathcal T}}

\newtheorem{theorem}{\bfseries Theorem}
\newtheorem{proposition}{\bfseries Proposition}
\newtheorem{lemma}{\bfseries Lemma}
\newtheorem{ex}{\sc Example}
\begin{document}
\begin{center}
{\bf\large\sc 
A relation for the Jones--Wenzl projector and \\ 
tensor space representations of the Temperley--Lieb algebra  
} 

\vspace*{4mm}
{ Andrei Bytsko   }
\vspace*{3mm}

\end{center}
\vspace*{2mm}

\begin{abstract}
A relation for the Jones-Wenzl projector is proven. It  has the following
consequence for representations of the Temperley-Lieb algebra on tensor product spaces:
if such a representation is built from a Hermitian $n\,{\times}\,n$ matrix $T$ of 
rank $r$ such that $T^2=Q T$, then either  $n^2 = Q^2 r$ and $Q^2 =1,2,3$
or $n^2 \geq 4 r$. For the latter class of representations, new examples are found.
This includes explicit examples for $r=2,3,4$ and any  $n \geq r$ (with one exception) and
a solution for $n=r+1$ with arbitrary~$r$.
\end{abstract}

\section{Introduction}

The Temperley--Lieb  algebra $TL_N(Q)$, where 
$Q >0$ and $N \geq 2$, is a unital associative algebra over $\mathbb C$
with generators $\ST_1,\,{\ldots}\,,\ST_{N-1}$ and relations 
\begin{align}
\label{TL1q}
{}& \ST_k^2 = Q \, \ST_k ,
 &&  \!\!\!\!\!\!\!\!\!\!\! \text{\rm for all}\ \ k  , \\
 \label{TL1}
{}& \ST_k \, \ST_m = 
\ST_m \, \ST_k  , &&
 \!\!\!\!\!\!\!\!\!\!\!  \text{\rm if} \ \ |k-m| \geq 2 \,,  \\
\label{TL2}
{}&  \ST_k \, \ST_{m}  \ST_k = \ST_k , 
 &&  \!\!\!\!\!\!\!\!\!\!\!  \text{\rm if} \ \ |k-m|=1  \,.
\end{align} 

Let $M_n$ be the ring of $n\,{\times}\,n$ complex matrices,
$I_n \in M_n$ denote the identity matrix,  and
$X^*$ stand for the  conjugate transpose of $X \in M_n$.
Consider a matrix $T \in M_{n^2}$ that satisfies the following relations:
\begin{align*} 
{}& (\mathrm{T1})     &&
  T^* = T ,&&  
	 T^2 = Q \, T   , \quad Q>0, && \\  
{}& (\mathrm{T2})     &&   
 T_{12} \, T_{23} \, T_{12} \, = T_{12}  ,&&  
 T_{23} \, T_{12} \, T_{23} \, = T_{23}  , &&    
\end{align*}
where 
$T_{12} \equiv T \,{\otimes}\, I_n$ and 
$T_{23} \equiv I_n \,{\otimes}\, T$, $\otimes$ 
denotes the matrix Kronecker product.

A solution to (T1)--(T2) defines a unitary (w.r.t. the involution $\ST_k^*=\ST_k$) 
representation $\tau$ of  $TL_N(Q)$ on the tensor product space
$\bigl({{\mathbb C}^n}\bigr)^{\otimes N}$  such that 
\begin{align}\label{tau}
 \tau(\ST_k) = T_{k,k+1} \equiv
I_n^{\otimes (k-1)} \otimes T \otimes I_n^{\otimes (N-k-1)} .
\end{align}
Every solution $T$ to (T1)--(T2) can be used to build an R-matrix,
i.e., a solution to the Yang-Baxter equation that
plays a key role in constructing quantum integrable models,
\begin{align}\label{YB}
   R_{12}(\lambda) R_{23}(\lambda+\mu) R_{12}(\mu) = 
   R_{23}(\mu) R_{12}(\lambda+\mu) R_{23}(\lambda).
\end{align}
Namely, $R(\lambda) = \sinh(\lambda + \gamma) I_{n^2} - \sinh (\lambda) \,T$,
where  $ e^{\gamma}+e^{-\gamma}=Q$, satisfies  equation~(\ref{YB}).
The corresponding constant R-matrix, $R = e^{\gamma} I_{n^2} - T$, 
satisfies the constant version of (\ref{YB}) and thus 
yields a braid group representation.
The most known example of this type is the R-matrix related to the 
fundamental representation of $U_q(\mathfrak{sl}_2)$, 
see also \cite{BK} for  solutions to (T1)--(T2) related in a similar way 
to some other quantum groups. 

Three important characteristics of a solution 
to (T1)--(T2) are its ``size'' $n$, the value of $Q$, and the rank  
$r=\mathrm{rank}(T)$. Note that these parameters are not independent.
In particular, the system (T1)--(T2) has no solution if $Q\, r < n$, see~\cite{By1}.

In the $r=1$ case, a solution to (T1)--(T2) exists if and only if $Q=n=1$ or 
$Q \geq n \geq 2$ and a general solution is known (cf.\ Corollary 1 in \cite{By2}).
But in the higher rank case, $r \geq 2$, solutions to (T1)--(T2) have been constructed
so far only in few particular cases and no general sufficient condition of existence 
of a solution to (T1)--(T2) with given values of $n$, $r$, and $Q$ seems to be know at present
(except for $Q=2$, in which case a necessary and sufficient condition on $n$ and $r$ has been 
found recently in \cite{LPW}, see also Section~\ref{QQ2} below).  

In the present article, we will establish a relation for the Jones--Wenzl projector that
will allow us to refine the necessary conditions of existence 
of a solution to (T1)--(T2). We will also construct some new varying $Q$
solutions to (T1)--(T2), where, for given $n$ and $r$, the value of $Q$ depends on a set
of parameters. This includes, in particular, the case $n=r+1$ that generalizes
the well-known $n=2$, $r=1$ solution related to the $U_q(\mathfrak{sl}_2)$ R-matrix.
We will show that one can add and multiply
in a certain sense two solutions to (T1)--(T2) if they have suitable ranks. Using 
these constructions, we will provide examples of solutions to (T1)--(T2) for
$r=2,3,4$ and any $n \geq r$ (with one exception). Finally, we will propose 
a conjectural refinement of a necessary condition of existence 
of a solution to (T1)--(T2).

\section{Jones--Wenzl projector and restrictions on  $Q$ and $n^2/r$}

Recall the definition of the Jones--Wenzl projector \cite{Jo1,We}. Fix $Q>0$.
Let $\rho_N, N=1,2,\ldots$ be the sequence of rational functions
in $Q$ defined inductively by
\begin{equation}\label{rhoind} 
  \rho_{N+1} = \bigl( Q - \rho_N \bigr)^{-1} , \qquad
  \rho_1=1/Q.
\end{equation}
 
Let $\SP_N, N=1,2,\ldots$ be the sequence of elements of 
$TL_N(Q)$  defined inductively by the relation (where $\SP_N$ is regarded 
as an element of $TL_N(Q) \subset TL_{N+1}(Q)$)
\begin{equation}\label{Pind} 
 \SP_{N+1} = \SP_N - \rho_N \, 
 \SP_N  \ST_{N} \, \SP_N , \qquad  \SP_1 = 1.
\end{equation} 
Note that $\SP_N$ is well defined if $\rho_k$  are finite for all $k < N$.

The key property of $\SP_N$ is  that it is an idempotent satisfying the following relations:
\begin{align}
\label{JW2}
{}& \ST_k \, \SP_N = \SP_N \, \ST_k = 0 , 
 \qquad  \text{for}\ \  k= 1,\ldots, N-1 .
\end{align}
Note that $\SP_N$ is invariant under the automorphism
$\phi(\ST_{k}) = \ST_{N-k}$ of $TL_{N}(Q)$.
Indeed, $\phi(\SP_N)$ satisfies the same relations (\ref{JW2}) and hence
$ \phi(\SP_N) =  \phi(\SP_N) \SP_N= \SP_N $.

The Temperley--Lieb algebra admits a normalized Markov trace~\cite{GHJ,Jo1}, i.e.,  
a linear map $\mathrm{Tr} : TL_N(Q) \to {\mathbb C}$ such that 
$\mathrm{Tr}(1)=1$, 
$\mathrm{Tr}(x y) = \mathrm{Tr}(y x)$ for all 
$x, y \in TL_{N}(Q)$,  and 
$\mathrm{Tr} (x \, T_{N}) = Q^{-1} \mathrm{Tr} (x)$ for any  
$x  \in TL_{N}(Q) \subset TL_{N+1}(Q)$. 
The requirement that this trace be positive, that is $\mathrm{Tr} (p) \geq 0$ 
if $p$ is an idempotent and, in particular, 
$\mathrm{Tr} (\SP_N) \geq 0$  for all $N$ such that $\SP_N$ is well defined, 
restricts the possible values of 
$Q$ to the range $Q \in J_\infty  \cup [2,\infty)$~\cite{Jo1}, where
\begin{equation}\label{Qnr2} 
    J_\infty \equiv \bigl\{ 2 \cos\bigl( \frac{\pi}{k+2}\bigl), \ 
   k=1,2,\ldots  \bigr\}.
\end{equation}  

Let a matrix $T \in M_{n^2}$ of rank $r$ be a solution to (T1)--(T2) and
$\tau$ be the corresponding tensor space representation of  $TL_N(Q)$ given by~(\ref{tau}).
Then $P_N = \tau(\SP_N)$ is an orthogonal projection, so its matrix trace must 
be non-negative for all $N$ such that $\SP_N$ is well defined. However, this
requirement does not imply that $Q \in J_\infty  \cup [2,\infty)$ as it was for the 
Markov trace. Indeed,  $\mathrm{Tr} (\SP_N)$ is a function in $Q$,
hence the above restriction on the range of~$Q$. But its counterpart  
$n^{-N} \tr (P_N)$ is a polynomial in $r/n^2$ that does not depend 
 on $Q$ at all (cf.\ eq.~(\ref{mtrnex})). Although the condition $\tr (P_N) \geq 0$
imposes some restrictions on $Q$ if $4r > n^2$ (cf.\ Theorem~4 in~\cite{By1}), 
it yields no restriction if $4r \leq n^2$ .

We will derive some restriction on the possible range of $Q$ in 
(T1)--(T2) from an observation that, in the representation (\ref{tau}), 
a certain matrix related to the Jones--Wenzl projector is positive semidefinite. 
To this end, we first need to establish a relation for the Jones--Wenzl projector. 

Let us regard $TL_N(Q)$  as a subalgebra of $TL_{N+2}(Q)$ and  denote  
\begin{align}\label{T'}
{}& \ST'_k = \ST_{k+1} , \qquad
  \ST''_k = \ST_{k+2}  \qquad  \text{for}\ \  k= 1,\ldots, N-1.
\end{align}
Let $\SP'_N$ and $\SP''_N$ stand for $\SP_N$, where 
each  $\ST_k$ is replaced, respectively,  by $\ST'_k$ and $\ST''_k$. 
For instance, $\SP_2 = 1 - Q^{-1} \ST_1$,  
$\SP'_2 = 1 - Q^{-1} \ST_2$, $\SP''_2 = 1 - Q^{-1} \ST_3$.

Let us commence with an observation that 
\begin{align}\label{PP2a}
(\SP_{2} - \SP_{2}' )^2 = \frac{\rho_1}{\rho_2} (1-\SP_3) .
\end{align}
Consider this relation in the representation (\ref{tau}). Let $Q \neq 1$.
Then $(P_{2} - P_{2}' )^2$ and $(I_{n^3}-P_3)=(I_{n^3}-P_3)^2$ are
nonzero positive semidefinite matrices (note that 
$\tr (I_{n^3}-P_3) = 2 nr$).  Which implies that
$\rho_1/\rho_2 = (Q^2 -1)/Q^2 > 0$ and hence $Q > 1$.

Although, for $N \geq 3$, $(\SP_{N} - \SP_{N}' )^2$ is not a multiple 
of a projection, we observe the following (proofs of all propositions are given
in the Appendix):
\begin{proposition}\label{HHrho}  
If $\rho_k$  are finite  for all $k < N$, then the following relation holds:
\begin{equation}\label{pp'} 
  (\SP_{N} - \SP_{N}' )^3 = 
  \frac{\rho_{N-1}}{\rho_{N}} \,  (\SP_{N} - \SP_{N}' ) .
\end{equation}  
\end{proposition}

 We will use relation (\ref{pp'}) in the proof of the following statement.

\begin{theorem}\label{QRan}
Every solution $T \in M_{n^2}$ of (T1)--(T2) that
has rank $r$ falls into one of the four classes: \ \  \\
a) $r=n^2$, $Q=1$  (the trivial solution $T=I_{n^2}$);\\ \ \ \  
b) $r=n^2/2$, $Q=\sqrt{2}$;\\ 
c) $r=n^2/3$, $Q=\sqrt{3}$;\ \ \ \\
d) $r \leq n^2/4$, $Q \geq \max (2, \frac{n}{r})$.
\end{theorem}

Thus, for a unitary tensor space representation of $TL_N(Q)$,
only three values of $Q$ from the set $J_\infty$ can occur,
and in general the range of $Q$ depends on $n$ and $r$.   
 
Let us mention  the following corollaries to Theorem~\ref{QRan}. 
\begin{proposition}\label{ColQR}
a) Apart from the trivial solution in the $n=2$, $r=4$ case, there exists no 
solution $T \in M_{n^2}$ to (T1)--(T2) of rank $r \leq 5$ such that $n<r$. \\
b)  The  constant R-matrix, $R = e^{\gamma} I_{n^2} - T$, where 
$T$ is a solution to (T1)--(T2) with $Q=e^{\gamma}+e^{-\gamma}$,
is unitary  only if $n^2=Q^2 r$ with  $Q \in \{1,\sqrt{2},\sqrt{3}\}$ 
or if $n^2= 4r +k^2$, $k \in {\mathbb Z}$ and $Q=2$.
\end{proposition}

\section{Sums and products of solutions}

Theorem~\ref{QRan} gives necessary conditions for 
existence of a unitary tensor space representation of $TL_N(Q)$. 
Clearly, a solution of the class a) exists for any $n$. For the classes b) and c),
examples of solutions are known (see Example \ref{ExQsq2} and Example \ref{ExQsq3}
below). A solution of the class d) for $r=1$, $n = 2$ with $Q$ varying in the range
$Q \in [2,\infty)$ is very well known (it admits a multi-parameter analogue
for $r=1$, $n \geq 2$ and $Q$ varying in the range $Q \in [n,\infty)$,
see Corollary 1 in \cite{By2} for its general form).
But for $r \geq 2$, explicit solutions to (T1)--(T2) have been constructed so far
only in few cases: solutions for $Q^2=r=n \geq 3$ were found in \cite{Kul3,WX},
families of solutions for $r=2$, $n=0 \pmod 3$ and $r=2$, $n=0 \pmod 4$ 
with $Q$ varying, respectively, in the range $Q \in [2n/3,\infty)$ and $Q \in [n/\sqrt{2},\infty)$
were  constructed in~\cite{By2}. 
In the rest of this article, we will provide some more explicit solutions to (T1)--(T2).

To this aim, we will first recall that the system (T1)--(T2) is equivalent to a condition
that certain partitioned matrix is unitary and then we will give three constructions that 
can be used to build new solutions to (T1)--(T2) from already known ones.

\subsection{Unitarity condition, examples with $Q<2$}

Let $e_1$, \ldots, $e_n$ be an orthonormal basis of ${\mathbb C}^n$. 
Then a matrix $V \in M_n$ defines a vector in 
${\mathbb C}^n \,{\otimes}\, {\mathbb C}^n$: 
$v = \sum_{a,b=1}^n V_{ab} \, e_a \otimes e_b$.
Similarly, a set of matrices $V_1$, \ldots, $V_r \in M_n$ such that
\begin{equation}\label{vv}
   \tr \bigl( V_s^*  V_m \bigr) 
  =\delta_{sm}  
\end{equation}
defines an orthonormal set of vectors  $v_1$, \ldots, $v_r$ in 
${\mathbb C}^n \,{\otimes}\, {\mathbb C}^n$
and thus the subspace $\CT$ spanned by these vectors
(we will write $\CT \sim \{V_1$, \ldots, $V_r\}$ always assuming that
condition (\ref{vv}) holds).
In the chosen basis, the orthogonal projection on $\CT$ is given by
\begin{equation}\label{PiTau}
  P_\CT = 
 \sum_{s=1}^r \sum_{a,b,c,d=1}^n (V_s)_{ab} \, 
 (\bar{V}_s)_{cd} \  E^{(n)}_{ac} \otimes E^{(n)}_{bd} ,
\end{equation}
where $\bar{V}$ denotes the complex conjugate of $V$ and 
$E^{(n)}_{ab} \in M_n$ denotes the matrix unit such that 
$(E^{(n)}_{ab})_{ij} = \delta_{ai}  \delta_{bj}$.

\begin{ex}\label{trivsol}
Let $\CT \sim \{V_1,\ldots, V_{n^2}\}$, where
$V_1=E^{(n)}_{11},\ldots, V_{n^2}=E^{(n)}_{nn}$ is the 
set of all matrix units in $M_n$. Then  $T=P_\CT=I_{n^2}$ is 
the trivial solution to (T1)--(T2) of rank $r=n^2$.  
\end{ex}

Given a subspace $\CT \sim \{V_1,\ldots, V_{r}\}$, $V_k \in M_n$, 
let us denote by $W_{\CT} \in M_{r n}$ be the following partitioned matrix
\begin{align}\label{WV} 
  W_{\CT}   
  = \sum_{s,m=1}^r  E^{(r)}_{sm} \otimes V_m \bar{V}_s \,.
\end{align} 
\begin{proposition}[\cite{By1}, Theorem~2]\label{Wuni}
$T=Q P_\CT$, where $\CT \sim \{V_1,\ldots,V_r\}$ and $Q > 0$,  is a solution to
(T1)--(T2) if and only if $Q W_\CT$ is a unitary matrix.
\end{proposition}

Let us list some known solutions to (T1)--(T2) in terms of matrices $V_1,\ldots,V_r$.

\begin{ex}\label{stansol}
For $n \geq 1$ and $z \in {\mathbb C} \setminus \{0\}$, let $V^{(n)} \in M_n$ 
and $Q_n(z)$ be given by
\begin{equation}\label{vnq}
V^{(n)} = \gamma_n \sum_{k=0}^{n-1} z^{k} E^{(n)}_{k+1,n-k} , \qquad
Q_n(z)  = \sum_{k=0}^{n-1} |z|^{2k+1-n} , 
\end{equation}
where $\gamma_n = \bigl( \sum_{k=0}^{n-1} |z|^{2k} \bigr)^{-1/2}$. 
Then $T= Q_n(z) P_{\CT} \in M_{n^2}$, where $\CT \sim \{V^{(n)}\}$, is a solution to (T1)--(T2)
of rank $r=1$. The expression for $Q_n(z)$ readily implies that $Q_n(z) \in [n,\infty)$.
The corresponding constant R-matrix is related to the highest weight 
representation of $U_q(\mathfrak{sl}_2)$ of weight $(n-1)$, cf.\ \cite{BK}.
\end{ex}
 
\begin{ex}\label{ExQsq2}
 Let   $V_1, V_2 \in M_2$ be given by
\begin{equation}\label{qrt2}
 V_1 = \frac{1}{\sqrt{2}} \bigl( E^{(2)}_{11} + E^{(2)}_{22} \bigr) , \qquad
 V_2 = \frac{1}{\sqrt{2}} \bigl( \sqrt{-1} \, E^{(2)}_{12} + E^{(2)}_{21} \bigr) .
 \end{equation}
Then $ T = Q P_{\CT} \in M_4$, where $\CT \sim \{V_1, V_2\}$ and $Q=\sqrt{2}$, is 
a solution to (T1)--(T2) of rank $r=2$. The corresponding constant R-matrix was 
found in~\cite{Hi1}.
\end{ex}

\begin{ex}\label{ExQsq3} 
Set $q = e^{2 \pi \sqrt{-1} /3}$.
Let   $V_1, V_2, V_3 \in M_3$ be given by
\begin{equation}\label{qrt3}
\begin{aligned}
 V_1 = \frac{1}{\sqrt{3}} \bigl( E^{(3)}_{13} + E^{(3)}_{22} + E^{(3)}_{31}\bigr) , & \quad
 V_2 = \frac{1}{\sqrt{3}} \bigl( q \, E^{(3)}_{12} + E^{(3)}_{21} + E^{(3)}_{33}\bigr),  \\
 V_3 = \frac{1}{\sqrt{3}}  \bigl( E^{(3)}_{11} & +    \bar{q} \, E^{(3)}_{23} + E^{(3)}_{32} \bigr). 
\end{aligned}
\end{equation}
Then $ T = Q P_{\CT} \in M_9$, where $\CT \sim \{V_1, V_2, V_3\}$ and $Q=\sqrt{3}$, is 
a solution to (T1)--(T2) of rank $r=3$. This solution was used in \cite{GJ}, see also \cite{WX}.
\end{ex}

To conclude this section, we recall that if $\CT \sim \{V_1,\ldots,V_r\}$, $V_k \in M_n$,  
defines a solution to (T1)--(T2) and $g \in M_n$ is a unitary matrix,
then $\CT' \sim \{V'_1,\ldots,V'_r\}$,
where $V_k' = g V_k g^t$ for $k=1,\ldots,r$, defines a unitarily equivalent solution.

\subsection{Sums and products of solutions}

For $X \in M_{n}$, $Y \in M_{m}$, let
$X \oplus Y \in M_{n+m}$ denote the block diagonal matrix with blocks $X$, $Y$.
Given two solutions to (T1)--(T2) of the same rank, we can construct their
``direct sum'' in the following sense.

\begin{proposition}\label{Tadd}
Let $\CT_1 \sim \{V^{(1)}_1, \ldots, V^{(1)}_r\}$, where 
$V^{(1)}_k \in M_{n_1}$, and 
$\CT_2 \sim \{V^{(2)}_1, \ldots, V^{(2)}_r\}$, where $V^{(2)}_k \in M_{n_2}$.
Suppose that $T_1=Q_1 P_{\CT_1} \in M_{n_1^2}$ and 
$T_2=Q_2 P_{\CT_2} \in M_{n_2^2}$  satisfy (T1)--(T2) for 
some $Q_1 , Q_2 > 0$. 
Set $\tilde{\CT} \sim \{\tilde{V}_1$, \ldots, $\tilde{V}_r\}$, 
where $\tilde{V}_k \in M_{n_1+n_2}$ is given by 
\begin{align}\label{dirsumV} 
 \tilde{V}_k = \frac{1}{ \sqrt{Q_1+Q_2}} 
( \sqrt{Q_1} V^{(1)}_k) \oplus ( \sqrt{Q_2} V^{(2)}_k) .
\end{align} 
Then 
$\tilde{T}=(Q_1 + Q_2) P_{\tilde{\CT}} \in  M_{(n_1+n_2)^2}$ 
is a solution to (T1)--(T2) of rank $r$.
\end{proposition}

\begin{ex}\label{Q2dirsum}
Let $n$ be even and
$V_1=E^{(\frac{n}{2})}_{11},\ldots, V_{(\frac{n}{2})^2}=E^{(\frac{n}{2})}_{\frac{n}{2}\frac{n}{2}}$
be the set of all matrix units in $M_{\frac{n}{2}}$. Set 
$\tilde{V}_k = \frac{1}{\sqrt{2}} V_k \oplus V_k$.
Then $T =  Q P_{\CT} \in M_{n^2}$, where $Q=2$ and 
$\CT \sim \{\tilde{V}_1,\ldots,\tilde{V}_{(\frac{n}{2})^2}\}$,  
is a solution to (T1)--(T2) of rank $r=n^2/4$. 
Thus,  even if both $T_1$ and $T_2$ in Proposition~\ref{Tadd} are trivial solutions, 
their ``direct sum'' is not a trivial solution.
\end{ex}

If $T_1$ and $T_2$ are solutions to (T1)--(T2) and one of them has rank one, 
we can construct their product in the following sense (we use here
the rank one solution defined in Example~\ref{stansol} but it obviously can be
replaced with any other solution of rank one). 

\begin{proposition}\label{VxU}
For $n, m \geq 2$ and $z \in {\mathbb C} \setminus \{0\}$, 
let $V^{(n)}$ and $Q_n(z)$ be given by (\ref{vnq}) and let 
$\CT \sim \{ V_1,\ldots,V_r\}$, where  $V_k \in M_{m}$. 
Suppose that $T= Q P_\CT \in M_{m^2}$ satisfies (T1)--(T2) for some $Q > 0$.
Let $\tilde{V}_k \in M_{mn}$ and $\tilde{Q}(z)$ be given  by
\begin{equation}\label{vmultun} 
\tilde{V}_k = V_k \otimes V^{(n)}  , \qquad  
  \tilde{Q}(z) = Q_n(z) \, Q .
\end{equation} 
Then $\tilde{T} = \tilde{Q}(z) P_{\tilde{\CT}} \in M_{m^2 n^2}$, where 
$\tilde{\CT} \sim \{\tilde{V}_1,\ldots,\tilde{V}_r\}$,  
is a solution to (T1)--(T2) of rank~$r$.
\end{proposition}

Note that the range of $Q_n(z)$ in (\ref{vmultun})  is $Q_n(z) \in [n\, Q, \infty)$.

\begin{ex}\label{Q2dirprod}
Let $n$ be even and
$V_1=E^{(\frac{n}{2})}_{11},\ldots, V_{(\frac{n}{2})^2}=E^{(\frac{n}{2})}_{\frac{n}{2}\frac{n}{2}}$
be the set of all matrix units in $M_{\frac{n}{2}}$. 
For $z \in {\mathbb C} \setminus \{0\}$, 
let $\tilde{V}_k \in M_{n}$ and $Q(z)$ be given by 
\begin{equation}\label{vunq2} 
  \tilde{V}_k = \frac{1}{\sqrt{ 1 + |z|^2 }} \, 
  V_k \otimes ( E^{(2)}_{12} + z E^{(2)}_{21} ) , \qquad
  Q(z) = |z| + 1/|z| .
\end{equation} 
Then $T =  Q(z) P_{\CT} \in M_{n^2}$, where
$\CT \sim \{\tilde{V}_1,\ldots,\tilde{V}_{(\frac{n}{2})^2}\}$,  
is a solution to (T1)--(T2) of rank $r=n^2/4$ with $Q$ varying in the range 
$Q(z) \in [2, \infty)$.
\end{ex}

Let us compare multiplication of solutions to (T1)--(T2) in the sense of Proposition~\ref{VxU}
  with the following folklore construction that can be called
fusion by analogy with a similar construction of solutions to the Yang-Baxter equation. 

\begin{proposition}\label{Fus}  
If $T \in M_{n^2}$ has rank $r$ and 
satisfies (T1)--(T2) for some  $Q>0$, then  
$\tilde{T} = T_{23} T_{12} T_{34} T_{23}  \in M_{n^4}$ has rank $\tilde{r}=r^2$ 
and satisfies (T1)--(T2) for $\tilde{Q} = Q^2$.
\end{proposition}

 In Proposition~\ref{Fus}, the ranks of $T$ and $\tilde{T}$ are different unless $r=1$,
whereas Propositions~\ref{Tadd} and~\ref{VxU} allow us to build new solutions from
known solutions of the same rank.

\section{Some new explicit solutions}

\subsection{The case of $n=r \leq 4$ and generalized permutation matrices}\label{Nr4}

One can notice that the solutions for $n=r$ with $r=1,2,3$ 
(cf.\ Examples \ref{ExQsq2} and \ref{ExQsq3}) are given by generalized permutation matrices,
i.e., $V_k = D_k P_{\sigma_k}$, where $D_k$ is a non-degenerate diagonal
matrix and $P_{\sigma_k}$ is the permutation matrix corresponding to an
element $\sigma_k$ of the symmetric group $S_n$. 
(This is the general form of a solution for $r=1$, $n \geq 2$, cf.\ \cite{By2}. 
For some varying $Q$ solutions of this type for $r=2$ see also~\cite{By2}.)
 
Let us extend this list with a solution of the same type for $r=n=4$.

\begin{proposition}\label{Solnr4}
Let $V_k \in M_4$, $k=1,2,3,4$ be given by 
\begin{equation}\label{soln4r4} 
\begin{aligned}
 V_1 &=   
  \bigl( z_1  E^{(4)}_{12} + z_2  E^{(4)}_{23} 
  + z_3  E^{(4)}_{34} + z_4  E^{(4)}_{41}\bigr) , \quad 
  V_2 =   
  \bigl( z_1  E^{(4)}_{14} + z_2  E^{(4)}_{21} + z_3  E^{(4)}_{32} + z_4  E^{(4)}_{43}\bigr) , \\[0.5mm]
 V_3 &=   \bigl( \bar{z}_3 E^{(4)}_{12} + \bar{z}_4 E^{(4)}_{21} - 
  \bar{z}_1 E^{(4)}_{34} - \bar{z}_2 E^{(4)}_{43} \bigr) , \quad
 V_4 =   \bigl(  \bar{z}_3 E_{14} + \bar{z}_4 E_{23}
  - \bar{z}_1 E_{32} - \bar{z}_2 E_{14} \bigr)  ,
\end{aligned}  
\end{equation}
where $z_1, z_2, z_3, z_4 \in {\mathbb C}$ are such that
\begin{equation}\label{zz4r4} 
 |z_1|^2 + |z_2|^2 + |z_3|^2 + |z_4|^2 =1 , \qquad
(|z_1| + |z_3|)(|z_2| + |z_4|) \neq 0 .
\end{equation}
 Then
$T= Q_{z_1,z_2,z_3,z_4} P_{\CT} \in  M_{16}$,
where ${\CT} \sim \{ V_1, V_2, V_3, V_4\}$ and 
\begin{equation}\label{qvrn4}
Q_{z_1,z_2,z_3,z_4} =  \frac{1}{\sqrt{(|z_1|^2 + |z_3|^2)(|z_2|^2 + |z_4|^2)}}
\end{equation} 
is a solution to (T1)--(T2) of rank $r=4$ with $Q$ varying in the range 
$Q_{z_1,z_2,z_3,z_4} \in [2,\infty)$.
\end{proposition}
 
Let us remark that, with the help of either Proposition~\ref{Tadd} (by taking $k$ copies
of (\ref{soln4r4})) or Proposition~\ref{VxU} (multiplying (\ref{soln4r4})  
with a rank one solution, say $V^{(k)}$), one can build from   
(\ref{soln4r4}) a solution for $r=4$ and $n=4k$, $k \in \mathbb N$, where $V_k$ are  
given by generalized permutation matrices as well. The corresponding 
value of $Q$ will be in the range $[2k, \infty)$.

\subsection{The case of $n=r+1$ with arbitrary $r$}

Let us give a solution for $n=r+1$, $r\in \mathbb N$,
which is not obtained from some smaller size solutions 
with the help of Proposition~\ref{Tadd} or Proposition~\ref{VxU}. This solution
is remarkably sparse -- every matrix $V_k$ has only two non-zero entries irrespective
of its size~$n$.

\begin{proposition}\label{Snrmo} 
For $n \geq 2$ and $z_1,z_2 \in {\mathbb C} \setminus \{0\}$, 
let $V_k \in M_n$, $k=1,\ldots,n-1$ be given by 
\begin{equation}\label{vtete1} 
 V_k = \frac{1}{\sqrt{ |z_1|^2 + |z_2|^2 }} \, 
  \bigl( z_1  E^{(n)}_{1,k+1} + z_2  E^{(n)}_{k+1,1} \bigr) .
\end{equation}  
Then 
$T= Q_{z_1,z_2} P_{\CT} \in  M_{n^2}$,
where ${\CT} \sim \{ V_1, \ldots, V_{n-1}\}$ and 
\begin{equation}\label{qvtete}
Q_{z_1,z_2} = |z_1|/|z_2| + |z_2|/|z_1|
\end{equation} 
is a solution to (T1)--(T2) of rank $r=n-1$ with $Q$ varying in the range 
$Q_{z_1,z_2} \in [2,\infty)$.
\end{proposition}

Let us remark that, for  $n=2$,  equation (\ref{vtete1}) yields the well-known
solution of rank one (cf.\ Example~\ref{stansol}) and, for  $n=3$, 
it recovers a particular case of a more general 
solution of rank two (cf.\ Theorem~3 in~\cite{By2}).

\subsection{The case of $r \leq 4$ with arbitrary $n$}

For $n <r \leq 4$, there is only one (trivial) solution to (T1)--(T2) (cf.\ the part a) of 
Proposition~\ref{ColQR}).
For $n=r  \leq 4$, examples of solutions were given in Example~\ref{ExQsq2},
Example~\ref{ExQsq3}, and Proposition~\ref{Solnr4}. 
For the remaining cases, $r \leq 4$, $n > r$, we have the following statement
(where the case of $r=1$ is omitted because it is covered by Example~\ref{stansol}).

\begin{theorem}\label{ListSmr}
For $r=2,3,4$ and $n > r$, 
a solution $T \in M_{n^2}$ of rank $r$  to (T1)--(T2)  exists 
for every $Q$ in the range $Q \in [Q_{r,n} , \infty)$, where  \\[0.5mm] 
a) $Q_{2,n}=  2(k-m)+m \sqrt{2} $ \  if $n$ is given by $n=3k-m$, 
$k \in \mathbb N$, $m=0,1,2$;\\[0.5mm] 
b) $Q_{3,n} = 2(k-m)+m \sqrt{3}$ \  if $n\neq 5$ and $n$ is given by 
$n=4k-m$, $k \in \mathbb N$, $m=0,1,2,3$;\\[0.5mm] 
c) $Q_{4,n} = 2k$ \  if $n$ is given by $n=5k-m$, $k \in \mathbb N$, $m=0,1,2$;\\[0.5mm] 
d) $Q_{4,n} = 2k-1$ \  if $n$ is given by $n=5k-m$, $k \in \mathbb N$, $m=3,4$.
\end{theorem}
 
In the proof given in the Appendix, we will build for these cases some explicit solutions 
to (T1)--(T2) from suitable small size solutions.

\section{Some remarks}

\subsection{On the case of $Q=2$}\label{QQ2}

A remarkable necessary and sufficient condition of existence of
a solution to (T1)--(T2) for $Q=2$ was found recently in~\cite{LPW}.

\begin{proposition}[\cite{LPW}, Proposition 6.3]\label{SolQ2}
For $Q=2$, a solution $T \in M_{n^2}$ of rank $r$  to (T1)--(T2)  exists
if and only if $\sqrt{n^2 -4r}$ is an integer.
\end{proposition}

Let us remark that the criterion given in Proposition~\ref{SolQ2} can be reformulated as follows.

\begin{proposition}\label{SolQ2b}
i) For $Q=2$, a solution $T \in M_{n^2}$ of rank $r$  to (T1)--(T2)   exists
if and only if $r$ has a divisor $m$ such that $n=m+ r/m$. \\
ii) In particular, if $r$ is a prime number, a solution $T \in M_{n^2}$ of rank $r$  to (T1)--(T2)  
for $Q=2$ exists if and only if $n=r+1$.
\end{proposition}
 
Solutions to  (T1)--(T2) constructed in Sections 3 and 4 allow us to provide
the following explicit examples to Proposition~\ref{SolQ2}.
 
\begin{ex}\label{ExQ2} 
Setting $|z_1|=|z_2|$ in (\ref{vtete1}), we obtain a family of solutions  to (T1)--(T2) for
$Q=2$, $n=r+1$, $r \in \mathbb N$.
Setting $|z|=1$ in (\ref{vunq2}), we obtain a family of solutions  to (T1)--(T2) for
$Q=2$, $r=n^2/4$, $n/2 \in \mathbb N$.  
Setting $|z_1|^2 + |z_3|^2 = |z_2|^2 + |z_4|^2 =1/2$ in (\ref{soln4r4}), 
we obtain a family of solutions  to (T1)--(T2) for $Q=2$, $r=n=4$.
\end{ex}
 
Let us remark that Example~\ref{ExQ2} provides examples of explicit  solutions to (T1)--(T2) 
for $Q=2$ for all the cases allowed by Proposition~\ref{SolQ2} when $r \leq 5$. 
Indeed, by Proposition~\ref{SolQ2b}, we have $n=r+1$ if $r=1,2,3,5$.
For $r=4$, the divisors of $r$ are $m=1,2,4$ and thus we have either $n=4$ or $n=5$.

\subsection{On the lower bound for $Q$}\label{LBQ}

By Theorem 3 of \cite{By1}, if $T \in M_{n^2}$ is a solution to (T1)--(T2) of rank $r$,
then  we have an estimate $Q \geq n/r$ for the corresponding value of $Q$ in~(T1). 
For $r=1$, this estimate is sharp (cf.\ Example~\ref{stansol}) 
but, for $r>1$, it probably can be improved.
In this context, it is worth to remark that for all the solutions to (T1)--(T2)
mentioned in this article we have
\begin{equation}\label{Qrn2}
 Q \geq \frac{2n}{r+1} .
\end{equation} 
Indeed, for the cases a), b), c) in Theorem \ref{ListSmr}, we have 
$Q_{r,n} = (2n + m \sqrt{r} (\sqrt{r} -1)^2)/(r+1)$. Note that this
formula applies also to the Example~\ref{stansol}, where $r=1$ and $Q \geq n$.
For the case d) in Theorem \ref{ListSmr}, we have 
$Q_{4,n} =  (2n + 2m -r-1)/(r+1)$, where $2m > r+1$. 
For the Example~\ref{ExQsq2}, 
Example~\ref{ExQsq3}, and solutions constructed in \cite{Kul3,WX},
we have $Q=\sqrt{n}$, $r=n \geq 1$. 
In this case, (\ref{Qrn2}) holds because $2 \sqrt{n} \leq n+1$ if $n \geq 1$.
Also, inequality (\ref{Qrn2}) holds obviously
for the Example~\ref{trivsol}, Example~\ref{Q2dirprod},  solution (\ref{soln4r4}),
and the rank two solution with $Q \geq n/\sqrt{2}$ constructed in \cite{By2}.
For any $Q=2$ solution allowed by Proposition~\ref{SolQ2}, 
inequality (\ref{Qrn2}) holds because, by Proposition~\ref{SolQ2b}, we have 
$n= m + r/m \leq r+1$. Finally, 
for the Example~\ref{stansol}, solution (\ref{vtete1}), and the rank two 
solution with $Q \geq 2n/3$ constructed in \cite{By2}, the estimate  
(\ref{Qrn2}) is sharp. 
At present, the author is not aware of any solution to (T1)--(T2) for which
(\ref{Qrn2}) does not hold. Thus, it is natural to conjecture
that the system (T1)--(T2) has no solution if $Q\, (r+1) < 2n$.

\section*{Appendix}

\begin{lemma}\label{TPrel} 
If $\rho_k$  are finite  for all $k < N$, then  the following relations hold
\begin{eqnarray}\label{tpt1}
& \ST_N \, \SP_{N} \, \ST_{N}  = 
 \frac{1}{\rho_N} \ST_N \, \SP_{N-1} ,&\\[0.5mm]
\label{tpt2}
& \ST_{N+1} \, \SP'_{N} \, \ST_{N+1}  
= \frac{1}{\rho_N} \ST_{N+1} \, \SP'_{N-1} ,  \quad
\ST_{1} \, \SP'_{N} \, \ST_{1}  
= \frac{1}{\rho_N} \ST_{1} \, \SP''_{N-1} , &\\[0.5mm]
\label{tpt3}
& \ST_{1} \, \SP'_{N} \, \ST_{N+1}\, \SP'_{N} \, \ST_{1} 
 = - \frac{1}{\rho_N \rho_{N+1}}  \ST_{1} \, \SP''_{N} +
 \frac{1}{\rho^2_N}  \ST_{1} \, \SP''_{N-1} , & \\[0.5mm]
 \label{tpt4}
& \ST_{N+1} \, \SP'_{N} \, \ST_{1}\, \SP'_{N} \, \ST_{N+1} 
 = - \frac{1}{\rho_N \rho_{N+1}}  \ST_{N+1} \, \SP_{N} +
 \frac{1}{\rho^2_N}  \ST_{N+1} \, \SP'_{N-1} .& 
\end{eqnarray}
\end{lemma} 

\begin{proof}[\bf Proof of Lemma~\ref{TPrel}]
If $\rho_k$  are finite  for all $k < N$, projectors $\SP_1,\ldots,\SP_N$ are well defined.
Relation (\ref{tpt1}) is well known and easily derived. 
Note that $\ST_{N}$ commutes with $\SP_{N-1}$. Hence
\begin{align*}
{} \ST_N \, \SP_{N} \, \ST_{N}  & \stackrel{(\ref{Pind})}{=} 
 \ST_N \, ( \SP_{N-1} - \rho_{N-1} \, 
 \SP_{N-1}  \ST_{N-1} \, \SP_{N-1} ) \, \ST_{N} \\
{}& \stackrel{(\ref{TL1q}),(\ref{TL2})}{=} 
 \ST_N \, \SP_{N-1} \, (Q - \rho_{N-1}) 
 \stackrel{(\ref{rhoind})}{=}  
   \frac{1}{\rho_N} \ST_N \, \SP_{N-1} .
\end{align*} 
The first relation in (\ref{tpt2}) is obtained from (\ref{tpt1})
by the shift (\ref{T'}). Taking into account the remark made after eq.~(\ref{JW2}),
we obtain the second relation in (\ref{tpt2}) from the first one by applying 
the automorphism
$ \phi(\ST_{k}) = \ST_{N+2-k}$ of $TL_{N+2}(Q)$.

Next, applying the automorphism
$\phi(\ST_{k})= \ST_{N+1-k}$ of $TL_{N+1}(Q)$
to (\ref{Pind}), we obtain another form of the inductive relation for $\SP_N$,
\begin{equation}\label{Pind2} 
 \SP_{N+1} = \SP'_N - \rho_N \, 
 \SP'_N  \ST_{1} \, \SP'_N .
\end{equation} 
Taking into account
that $\ST_1$ commutes with $\SP''_{N-1}$ and $\SP''_{N}$, we
verify (\ref{tpt3}):
\begin{align*}
{}& \ST_{1} \, \SP'_{N} \, \ST_{N+1}\, \SP'_{N} \, \ST_{1} 
 \stackrel{(\ref{Pind})}{=} \frac{1}{\rho_N}
 \ST_{1} \, ( \SP'_{N} - \SP'_{N+1}) \, \ST_{1} \\
{}&   \stackrel{(\ref{Pind2})}{=}
\frac{1}{\rho_N} \ST_{1} \, (  \SP''_{N-1} - \rho_{N-1} \, 
 \SP''_{N-1}  \ST_{2} \, \SP''_{N-1} - \SP''_{N}
  + \rho_{N} \, 
 \SP''_{N}  \ST_{2} \, \SP''_{N}) \, \ST_{1} \\
{}& \stackrel{(\ref{TL1q}),(\ref{TL2})}{=} 
 \frac{1}{\rho_N} \ST_{1} \, (  \SP''_{N-1} (Q- \rho_{N-1})  
 - \SP''_{N} (Q - \rho_{N} )  ) 
 \stackrel{(\ref{rhoind})}{=}   
 \frac{1}{\rho_N} \ST_{1} \, (  \frac{1}{\rho_N} \SP''_{N-1}  
 - \frac{1}{\rho_{N+1}} \SP''_{N}   )   .
\end{align*}
Relation (\ref{tpt4}) is obtained from (\ref{tpt3}) by the automorphism
$\phi(\ST_{k}) = \ST_{N+2-k}$ of $TL_{N+2}(Q)$.
\end{proof}

\begin{proof}[\bf Proof of Proposition~\ref{HHrho}]
First, we note that 
\begin{align}
{}  \SP_{N+1} - \SP'_{N+1}  
 & \stackrel{(\ref{Pind}),(\ref{Pind2})}{=}   
 \label{pmp1}
\rho_N \, \SP'_N \, (\ST_{N+1} - \ST_{1}) \, \SP'_N .
\end{align} 
Relations (\ref{JW2}) imply that 
\begin{align}\label{ppnil}
\SP'_{N-1} \, \SP'_N = \SP'_N \, \SP'_{N-1} = 
\SP''_{N-1} \, \SP'_N = \SP'_N \, \SP''_{N-1} =\SP'_N . 
\end{align} 
Therefore, 
\begin{align}
\nonumber
{}& ( \SP_{N+1} - \SP'_{N+1}  )^2  \stackrel{(\ref{pmp1})}{=} 
  \rho^2_N \, \SP'_N \, (\ST_{N+1} - \ST_{1}) \SP'_N 
  (\ST_{N+1} - \ST_{1}) \, \SP'_N  \\
\nonumber
 {}&  \stackrel{(\ref{tpt2})}{=}
 \rho^2_N \, \SP'_N \, (\frac{1}{\rho_N} \ST_{N+1} \, \SP'_{N-1} 
 + \frac{1}{\rho_N} \ST_{1} \, \SP''_{N-1} 
 - \ST_{1} \, \SP'_N \, \ST_{N+1} - \ST_{N+1} \, \SP'_N \, \ST_{1} ) \, \SP'_N \\
\label{pmp2} 
 {}&  \stackrel{(\ref{ppnil})}{=}
 \rho^2_N \, \SP'_N \, (\frac{1}{\rho_N} \ST_{N+1}  
 + \frac{1}{\rho_N} \ST_{1} 
 - \ST_{1} \, \SP'_N \, \ST_{N+1} - \ST_{N+1} \, \SP'_N \, \ST_{1} ) \, \SP'_N .
\end{align} 
Combining (\ref{pmp1}) with (\ref{pmp2}), we derive relation (\ref{pp'}): 
\begin{align*}
{}& ( \SP_{N+1} - \SP'_{N+1}  )^3  = \\
{}& \stackrel{(\ref{pmp1}),(\ref{pmp2})}{=}
 \rho^3_N \, \SP'_N \, (\frac{1}{\rho_N} \ST_{N+1}  
 + \frac{1}{\rho_N} \ST_{1} 
 - \ST_{1} \, \SP'_N \, \ST_{N+1} - \ST_{N+1} \, \SP'_N \, \ST_{1} ) \, \SP'_N
 \, (\ST_{N+1} - \ST_{1}) \, \SP'_N \\
{}& \stackrel{(\ref{tpt2}),(\ref{tpt3}),(\ref{tpt4})}{=} 
 \rho^3_N \, \SP'_N \, (\frac{1}{\rho^2_N} \ST_{N+1} \, \SP'_{N-1} 
 + \frac{1}{\rho_N} \ST_{1} \, \SP'_N \, \ST_{N+1} 
 - \frac{1}{\rho_N} \ST_{1} \, \SP'_N \, \ST_{N+1} \, \SP'_{N-1} \\
{}& \qquad  +  \frac{1}{\rho_N \rho_{N+1}}  \ST_{N+1} \, \SP_{N} -
 \frac{1}{\rho^2_N}  \ST_{N+1} \, \SP'_{N-1}  
 - \frac{1}{\rho_N} \ST_{N+1} \, \SP'_N \, \ST_{1} 
 - \frac{1}{\rho^2_N} \ST_{1} \, \SP''_{N-1} \\
{}& \qquad  - \frac{1}{\rho_N \rho_{N+1}}  \ST_{1} \, \SP''_{N} +
 \frac{1}{\rho^2_N}  \ST_{1} \, \SP''_{N-1} 
 + \frac{1}{\rho_N} \ST_{N+1} \, \SP'_N \, \ST_{1} \, \SP''_{N-1} ) \, \SP'_N    \\
%
{}&  \stackrel{(\ref{ppnil})}{=}
\rho^2_N \, \SP'_N \, (\frac{1}{\rho_N} \ST_{N+1}  
 +  \ST_{1} \, \SP'_N \, \ST_{N+1} 
 -  \ST_{1} \, \SP'_N \, \ST_{N+1}   
 +  \frac{1}{\rho_{N+1}}  \ST_{N+1} \, \SP_{N} -
 \frac{1}{\rho_N}  \ST_{N+1}    \\
{}& \qquad 
 - \ST_{N+1} \, \SP'_N \, \ST_{1} 
 - \frac{1}{\rho_N} \ST_{1}  - \frac{1}{\rho_{N+1}}  \ST_{1} \, \SP''_{N} +
 \frac{1}{\rho_N}  \ST_{1} 
 +  \ST_{N+1} \, \SP'_N \, \ST_{1}  ) \, \SP'_N     \\
%
{}& = \frac{\rho^2_N}{ \rho_{N+1}} 
\SP'_N \, ( \ST_{N+1} \, \SP_{N} -   \SP''_{N} \, \ST_{1}  ) \, \SP'_N  
\end{align*}\begin{align*}
{}&  \stackrel{(\ref{Pind}),(\ref{Pind2})}{=} 
\frac{\rho^2_N}{ \rho_{N+1}} 
\SP'_N \, \bigl( \ST_{N+1} \, (\SP'_{N-1} - \rho_{N-1}
 \SP'_{N-1} \, \ST_1 \, \SP'_{N-1})   \\
{}& \qquad\quad   -
  (\SP''_{N-1} - \rho_{N-1} \SP''_{N-1} \, \ST_{N+1} \, \SP''_{N-1})   
  \, \ST_{1}  \bigr) \, \SP'_N \\
{}&  \stackrel{(\ref{ppnil})}{=}
\frac{\rho^2_N}{ \rho_{N+1}} 
\SP'_N \, ( \ST_{N+1} \,  - \rho_{N-1} \ST_{N+1} \,
 \SP'_{N-1} \, \ST_1 - \ST_1  + 
 \rho_{N-1}  \ST_{N+1} \, \SP''_{N-1} \, \ST_1 ) \, \SP'_N \\
{}&  = \frac{\rho^2_N}{ \rho_{N+1}} 
\SP'_N \, ( \ST_{N+1} \,  - \rho_{N-1} 
 \SP'_{N-1} \, \ST_{N+1} \, \ST_1 - \ST_1  + 
 \rho_{N-1}  \ST_{N+1}  \, \ST_1 \, \SP''_{N-1} ) \, \SP'_N \\
 {}&  \stackrel{(\ref{ppnil})}{=}
 \frac{\rho^2_N}{ \rho_{N+1}} 
\SP'_N \, ( \ST_{N+1} \,  - \ST_1   ) \, \SP'_N  
\stackrel{(\ref{pmp1})}{=} 
 \frac{\rho_N}{ \rho_{N+1}} ( \SP_{N+1} - \SP'_{N+1}  ) .  
\end{align*} 
\end{proof} 

\begin{proof}[\bf Proof of Theorem~\ref{QRan}]
Recall that the Chebyshev polynomials of the second kind are given by 
$U_m(t)=  2^m \prod\limits_{k=1}^{ m} 
   \bigl(t -  \cos\bigl(\frac{\pi k}{m+1}\bigr) \bigr)$. 
They satisfy the   recurrence relation $U_{m+1}(t)= 2t \,U_m(t) - U_{m-1}(t)$.
For $Q \neq 2$,  the solution to (\ref{rhoind}) can be expressed as
\begin{equation}\label{rhonex} 
 \rho_N = \frac{U_{N-1}(Q/2)}{U_{N}(Q/2)} .
\end{equation}
 For $Q = 2$, we have $\rho_N = N/(N+1)$.

 Let $\tau$ be a unitary tensor space representation of  $TL_N(Q)$ defined according to (\ref{tau})
by a matrix $T \in M_{n^2}$ of rank $r$ that satisfies (T1)--(T2). 
Set $P_N = \tau(\SP_N)$. Then
\begin{align}
\nonumber
\tr (P_{N+1}) & \stackrel{(\ref{Pind})}{=} \tr (P_{N} \otimes I_n)
 - \rho_N \tr (T_{N,N+1} P_N) = n \tr (P_{N})
 - Q^{-1} \rho_N \tr (T_{N,N+1} P_N T_{N,N+1}) \\
\nonumber
{}& \stackrel{(\ref{tpt1})}{=} n \tr (P_{N}) - Q^{-1} \tr ( P_{N-1} T_{N,N+1})
= n \tr (P_{N}) - Q^{-1} \tr ( P_{N-1} \otimes T) \\[1mm]
\nonumber {}& =  n \tr (P_{N}) - r \tr ( P_{N-1}) .
\end{align}
Comparing this relation with the recurrence relation for the Chebyshev polynomials
and taking into account that $\tr(P_1) =n$, $\tr(P_2)=n^2-r$,
we infer that 
\begin{equation}\label{mtrnex} 
 \mathrm{tr}( P_{N})  =  r^{N/2} \, U_{N} \bigl( \frac{n}{2\sqrt{r}} \bigr) .
\end{equation}

Let us show that if $T \in M_{n^2}$  satisfies (T1)--(T2), then   
$Q \in J_\infty  \cup [2,\infty)$. Indeed, suppose that $Q<2$ but $Q \notin J_\infty$.
Then there is $N \geq 3$ such that
$\cos ( \frac{\pi}{N} ) < Q/2 < \cos  (\frac{\pi}{N+1} )$ 
(recall that $Q\geq 1$ as a consequence of relation (\ref{PP2a})). 
Since $ \cos (\frac{\pi}{N+1} )$ is the maximal root of $U_N(t)$, 
we have $U_m(Q/2) >0$ for all $m < N$ and so, by 
equation (\ref{rhonex}),  $\rho_k$  are finite and positive for all $k < N$. 
Therefore, $\SP_1,\ldots,\SP_N$ are well defined by (\ref{Pind}) and so we  
can consider relation (\ref{pp'}) in the representation~(\ref{tau}). 
Its immediate consequence  is the relation 
\begin{equation}\label{PPp4}
 (P_{N} - P_{N}' )^4 = 
  \frac{\rho_{N-1}}{\rho_{N}} \,  (P_{N} - P_{N}' )^2 .
\end{equation}
Note that if  $Q/2$ is in the range
$\cos ( \frac{\pi}{N} ) < Q/2 < \cos  (\frac{\pi}{N+1} )$, $N \geq 3$, 
then it lies between the maximal and the next to maximal root of $U_N(t)$
(this fact was used in \cite{Jo1} to show that the restriction $Q \in J_\infty  \cup [2,\infty)$
follows from the requirement $\mathrm{Tr}( \SP_{N}) \geq 0$). 
So, $U_N(Q/2) <0$ and thus, by (\ref{rhonex}), we have $\rho_N <0$.
Note that the r.h.s. of (\ref{PPp4}) is not zero.
Indeed, if it were, then
relation (\ref{pp'}) and the fact that $\rho_{N-1}/\rho_{N} \neq 0$ would 
imply that $P_{N} =P_{N}'$. But then we have $T \otimes P_N = 
T_{12}  P_N'' = T_{12}  P_N =0$. Hence $P_N=0$ and, 
taking into account that $\rho_k$  are finite for all $k \leq  N$ and 
using (\ref{tpt1}), we infer by induction that $P_k=0$ for all $k \leq  N$,
 which is impossible. Thus, we conclude that the l.h.s. of (\ref{PPp4}) is a 
nonzero positive semidefinite matrix and the r.h.s. of (\ref{PPp4}) is a nonzero 
negative semidefinite matrix (since $\rho_{N-1}/\rho_{N} < 0$).
This contradiction implies that 
either $Q \in J_\infty$ or  $Q \geq 2$.

First, consider the case $Q \geq 2$. 
Then $\rho_N >0$ for all $N$ 
and $\SP_N$ is well defined by (\ref{Pind}) for any $N$.  
Since $P_N$ is positive semidefinite, we have $\tr (P_N) \geq 0$ for all $N$.
Suppose that $n^2 < 4 r \leq 4 n^2$. Then there is $N \geq 3$ such that  
$\cos ( \frac{\pi}{N} ) \leq \frac{n}{2\sqrt{r}} < \cos  (\frac{\pi}{N+1} )$,
that is $\frac{n}{2\sqrt{r}}$ lies between the maximal and the next to maximal 
root of $U_N(t)$. Whence, by (\ref{mtrnex}), we have $\tr (P_N) <0$.
This contradiction implies that  
$Q \geq 2$ is possible only if $r \leq  n^2/4$ (and we will see below that 
$Q< 2$ is not possible if $r\leq n^2/4$). 
This covers the class d) of solutions in Theorem~\ref{QRan}.
For this class, Theorem~3 in~\cite{By1} imposes an additional restriction, $Q \geq n/r$.

Now, consider the case  $Q \in J_\infty$, that is  
$Q =2\cos ( \frac{\pi}{N+1} )$ for some $N \geq 2$.
In this case, $\rho_k$  are finite and positive for $k < N$ but $U_N(Q/2)=0$
and thus $\rho_{N} = \infty$. However, $\SP_1,\ldots,\SP_{N}$ are still well defined.
Observe that 
$\ST_N  \SP_{N} \ST_{N}  = 0$ (which is derived in the same way as 
relation (\ref{tpt1}) by taking into account that $\rho_N = \infty$ iff  
$\rho_{N-1} =Q$). Therefore, in the representation $\tau$, we have
$ (P_{N} T_{N,N+1}  )^* P_{N} T_{N,N+1} =T_{N,N+1}  P_{N} T_{N,N+1} = 0 $.
Which implies that $P_N T_{N,N+1}  =0$. Multiplying this relation from both
sides by $T_{N+1,N+2}$ (which commutes with $P_N$), 
we infer that $P_N T_{N+1,N+2} = P_N \otimes T = 0$. Thus, $P_N=0$.
Whence by (\ref{mtrnex}) we have $U_N\bigl(\frac{\sqrt{s}}{2 }\bigr)=0$, where
$s = n^2/r$. Therefore, 
$s  = 4 (\cos \frac{\pi l}{N+1})^2 = 
e^{\frac{2\pi i l}{N+1}} + 2  +e^{-\frac{2\pi i l}{N+1}}$ for some~$l \in [1,N]$.
Since $s$ is a sum of three algebraic integers, it is an algebraic integer.  
But any rational algebraic integer 
is an ordinary integer (cf.\ Theorem 206 in~\cite{HW}).
So, we conclude that $s \in \{ 1,2,3 \}$.
A direct inspection shows that, for $s=1,2,3$, the only value of $N$ such that
$U_N(\frac{\sqrt{s}}{2})=0$ and $U_m(\frac{\sqrt{s}}{2}) \geq 0$ for all $m < N$
($\SP_m$ are well defined for $m < N$, so we must have $\tr(P_m) \geq 0$)  is,
respectively, $N=2,3,5$. The corresponding values of $Q =2\cos ( \frac{\pi}{N+1} )$
are, respectively, $Q=1,\sqrt{2},\sqrt{3}$.
This covers the classes a), b), and c) of solutions in Theorem~\ref{QRan}. 
\end{proof} 

\begin{proof}[\bf Proof of Proposition~\ref{ColQR}]
a) The only values of $r$ and $n$ which satisfy the inequalities
$\sqrt{r} \leq n < r \leq 5$ are $r=3, n=2$; $r=4, n=2$; $r=4, n=3$;
$r=5, n=3$; and $r=5, n=4$. In all of these cases, we have $r >n^2/4 $
which implies that none of them can correspond to a solution of the class d)
in Theorem~\ref{QRan}. For the classes a), b), and c), we must have 
$n^2/r=s$ with $s=1,2,3$, respectively. Which holds only for the pair $r=4, n=2$.  

b) It is easy to see that $R^{-1} = e^{-\gamma} I_{n^2} - T$. 
Since $T$ is Hermitian, $R$ is unitary only if  $\bar{\gamma} = -\gamma$.
Hence $Q=e^{\gamma}+e^{-\gamma} \leq 2$. It remains to
invoke Theorem~\ref{QRan} in the case $Q<2$ and Proposition~\ref{SolQ2}
in the case $Q=2$.
\end{proof} 

\begin{proof}[\bf Proof of Proposition~\ref{Tadd}] 
By Proposition~\ref{Wuni}, $Q_1 W_{\CT_1}$ and $Q_2 W_{\CT_2}$ are unitary
which is equivalent to the following equations on $V_k^{(i)}$, $i=1,2$
 ($V^t$ denotes the transpose of $V$)
\begin{align}\label{vvvvI}  
{}&       
  Q_i^2 \sum_{s=1}^r V_s^{(i)} \bar{V}_l^{(i)}
   (V_p^{(i)})^t (V_s^{(i)})^* = \delta_{lp} I_{n_i} .
\end{align} 
Note that 
$\tilde{V}_k = \frac{1}{ \sqrt{Q_1+Q_2}} 
( \sqrt{Q_1} V^{(1)}_k) \oplus ( \sqrt{Q_2} V^{(2)}_k)$ satisfy~(\ref{vv})
and, since they are block diagonal, they satisfy (\ref{vvvvI}) with $Q_i$ replaced 
by $(Q_1+Q_2)$ and $n_i$ is  replaced by $n_1+n_2$. 
Therefore, $(Q_1+Q_2) W_{\tilde{\CT}}$ is 
unitary and so, again by Proposition~\ref{Wuni},
$\tilde{T}$ is a solution to (T1)--(T2).
\end{proof}

\begin{proof}[\bf Proof of Proposition~\ref{VxU}]
For $V^{(n)}$ given by (\ref{vnq}),  matrix $Q_n(z) V^{(n)} \bar{V}^{(n)}$ is unitary.
Therefore, it suffices to note that 
$W_{\tilde{\CT}}=W_{\CT} \otimes (V^{(n)} \bar{V}^{(n)})$
and use Proposition~\ref{Wuni}.  
\end{proof}

\begin{proof}[\bf Proof of Proposition~\ref{Fus}] 
For $\tilde{T}_{1234}=T_{23} T_{12} T_{34} T_{23}$, the first 
relation in (T1) is obvious since $T_{12}$ and $T_{34}$ commute.
The second relation in (T1) and the first relation in (T2) are verified directly:
\begin{align*} 
{}& \tilde{T}_{1234}^2 =  
Q\, T_{23} T_{12} T_{34} T_{23}   T_{12} T_{34} T_{23} 
= Q\, T_{23}   T_{34}   T_{12} T_{34} T_{23}  
= Q^2  T_{23}    T_{12} T_{34} T_{23} = Q^2 \tilde{T}_{1234} , \\
{}& \tilde{T}_{1234}  \tilde{T}_{3456}  
 \tilde{T}_{1234}   = T_{23} T_{12} T_{34} T_{23}  
 T_{45} T_{34} T_{56} T_{45}   
  T_{23} T_{12} T_{34} T_{23}  \\
{}&  \qquad\qquad\qquad\ \,  
=  T_{23} T_{12} T_{34}  T_{45} T_{23}   T_{56} T_{45}   
    T_{12} T_{34} T_{23} = T_{23} T_{12} T_{34} 
  T_{23}  T_{45}  T_{12} T_{34} T_{23} \\
{}& \qquad\qquad\qquad\ \,  
= T_{23}  T_{34}  T_{12}  T_{45}   T_{34} T_{23} 
 = T_{23}  T_{12}  T_{34} T_{23}  
= \tilde{T}_{1234} .
\end{align*}
The second relation in (T2) is checked analogously. 
Finally, we have 
$Q^2 \tilde{r} =\tilde{Q} \tilde{r} = \tr \tilde{T}_{1234} = Q \tr (T_{12} T_{34} T_{23})
 = \tr (T_{12} T_{34} T_{23} T_{12}) = \tr ( T_{34}  T_{12}) = 
 \tr (T \otimes T) = Q^2 r^2$. Whence $\tilde{r} = r^2$.
\end{proof} 

\begin{proof}[\bf Proof of Proposition~\ref{Solnr4}]
We have $V_k = D_k P_{\sigma_k}$, where $D_k$ is a diagonal
matrix and $P_{\sigma_k}$ is the permutation matrix corresponding to an
element $\sigma_k \in S_4$. 
Set $\Lambda_{abcd} = \mathrm{diag}(|z_a|^2, |z_b|^2,|z_c|^2,|z_d|^2)$.
Note that 
$V_1 V_1^* = V_2 V_2^* = \Lambda_{1234}$,
$V_3 V_3^* = V_4 V_4^* = \Lambda_{3412}$, and 
$V_k V_p^*$ are traceless matrices if $k\neq p$. 
Therefore, taking the first relation in (\ref{zz4r4}) into account, we see that
 $V_k$ satisfy relations (\ref{vv}). Using that
$V_k \Lambda_{abcd} = \Lambda_{\sigma_k(abcd)} V_k$,
one can check that $ \sum_{s=1}^4 V_s \bar{V}_k
   V_k^t  V_s^* =  (\Lambda_{2341} + \Lambda_{4123}) \Lambda_{1234}
   +  (\Lambda_{2143} + \Lambda_{4321}) \Lambda_{3412} = Q_{z_1,z_2,z_3,z_4}^{-2} I_{4} $.
In order to verify that $Q_{z_1,z_2,z_3,z_4} W_{\CT}$ is unitary it remains
to check that $\sum_{s=1}^4 V_s \bar{V}_k  V_p^t  V_s^* = 0$ if $k \neq p$
(cf.\ equation (\ref{vvvvI})) which can be done by a direct computation.
Thus, the claim follows by invoking Proposition~\ref{Wuni}.
Setting $\zeta=\sqrt{(|z_1|^2 + |z_3|^2)/(|z_2|^2 + |z_4|^2)}$ and 
taking the first relation in (\ref{zz4r4}) into account, we infer that
$Q_{z_1,z_2,z_3,z_4} = \zeta + 1/\zeta \geq 2$.
\end{proof}

\begin{proof}[\bf Proof of Proposition~\ref{Snrmo}]
Set $\gamma =  1/(|z_1|^2 + |z_2|^2)$.
Note that $V_k V_p^* = \gamma \delta_{kp} |z_1|^2 E^{(n)}_{11} + 
\gamma |z_2|^2 E^{(n)}_{k+1,p+1}$. Thus, $V_k$ satisfy relations (\ref{vv}). 
Substituting (\ref{vtete1}) in (\ref{WV}), we obtain
\begin{align}\nonumber 
{}&       
 W_{\CT} = W_1 + W_2 , \qquad
 W_1 =\gamma z_1 \bar{z}_2 \, I_r \otimes E^{(n)}_{11} , \quad 
 W_2 = \gamma z_2 \bar{z}_1 \sum_{s,m=1}^r E^{(r)}_{sm} \otimes E^{(n)}_{m+1,s+1} .
\end{align}  
Note that $W_1^t =  W_1$,  $W_2^t =  W_2$, and 
$W_1 W_2 = W_2 W_1 =0$. Therefore,
\begin{align}\nonumber 
 ( \gamma |z_1| |z_2|)^{-2} W_{\CT} W_{\CT}^*  & = 
 I_r \otimes E^{(n)}_{11} + 
  \sum_{m=1}^r I_r \otimes E^{(n)}_{m+1,m+1} = I_r \otimes I_n .
\end{align} 
Thus, $\gamma |z_1| |z_2| \,W_{\CT}$  is unitary
and the claim follows by invoking Proposition~\ref{Wuni}.
\end{proof}

\begin{proof}[\bf Proof of Theorem~\ref{ListSmr}]
We will use Proposition~\ref{Tadd} in order to construct ``direct sums" of
solutions to (T1)--(T2). 
Below, $T_{(r)}$ will denote solution (\ref{vtete1}) for a given~$r$,
$\tilde{T}^{(n)}$ will denote solution  (\ref{vnq}) for a given~$n$, 
 and $m$ will be an integer in the range 
$[0,\ldots,r]$.  We will refer to $n$ as the ``size" of a solution (although $T \in M_{n^2}$). 
Recall that, in this proof, $n>r$.

a) Let $n \neq 4$. Taking the sum of $m$ copies of the solution (\ref{qrt2})
and $(k-m)$ copies of the solution~$T_{(2)}$ we obtain, by Proposition~\ref{Tadd},  
a solution to (T1)--(T2) of size $n=3k-m$ and rank $r=2$ for any $Q \geq \sqrt{2} m + 2(k-m)$.
For $n=4$, we have $k=m=2$ and so this construction yields a 
solution only for $Q=2 \sqrt{2}$. However, taking the product, in the
sense of Proposition~\ref{VxU}, of the solution (\ref{qrt2}) with the solution 
$\tilde{T}^{(2)}$, we obtain a solution for $n=4$, $r=2$ and any $Q \geq 2 \sqrt{2}$
(another solution for  $n=4$, $r=2$ was given in \cite{By2}, Proposition~7). 

b) Let $n \neq 5, 6,9$. Taking the sum of $m$ copies of the solution (\ref{qrt3})
and $(k-m)$ copies of the solution~$T_{(3)}$, we obtain 
a solution to (T1)--(T2) of size $n=4k-m$ and rank $r=3$ for any $Q \geq \sqrt{3} m + 2(k-m)$.
Taking the product of the solution (\ref{qrt3}) with the solution 
$\tilde{T}^{(2)}$ or $\tilde{T}^{(3)}$, we obtain, respectively, 
a solution for $n=6$, $r=3$ and any $Q \geq 2 \sqrt{3}$ or  
$n=9$, $r=3$ and any $Q \geq 3 \sqrt{3}$.
For $n=5$, a solution cannot be given by a ``direct sum" because, by
Proposition~\ref{ColQR}, there exists no solution to (T1)--(T2) of rank 3 for $n=1,2$.

c)-d) Recall that (\ref{soln4r4}) 
provides a solution to (T1)--(T2) for $n=r=4$ and any $Q \geq 2$.
We will denote this solution by $T'$. 

Taking the sum of $m$ copies of the solution $T'$
and $(k-m)$ copies of the solution~$T_{(4)}$, we obtain 
a solution to (T1)--(T2) of size $n=5k-m$ and rank $r=4$ for any $Q \geq 2 m + 2(k-m) =2k$.
However, for $m=3,4$, we can obtain a solution with $Q \geq 2k -1$.
Indeed,  taking the sum of one copy of the trivial solution $T=I_4$, 
and $(k-1)$ copies of the solution~$T_{(4)}$, we obtain 
a solution to (T1)--(T2) of size $n=5k-3$ and rank $r=4$ for any $Q \geq 1 + 2(k-1) =2k-1$.
Also, taking the sum of one copy of the solution $T'$, one copy of the trivial
solution $T=I_4$, and $(k-2)$ copies of the solution~$T_{(4)}$,  we obtain
 a solution to (T1)--(T2) of size $n=5k-4$ and rank $r=4$ for any $Q \geq 2+ 1 + 2(k-2) =2k-1$.
\end{proof}  

\begin{proof}[\bf Proof of Proposition~\ref{SolQ2b}]
i) If $m$ divides $r$ and $n=m+r/m$, then $n^2 - 4r=(m-r/m)^2$,
so that the condition of Proposition~\ref{SolQ2} is fulfilled.
On the other hand, if the condition of Proposition~\ref{SolQ2} is fulfilled,
then $m=\frac{1}{2}(n + \sqrt{n^2 -4r})$ is an integer (note that $n$ and 
$\sqrt{n^2 -4r}$ have the same parity) and we have $r=m(n-m)$. 
Thus, $m$ divides $r$ and $n=m+r/m$. \\
ii) By i), we have $n=m+r/m$, where $m=1$ or $m=r$. Hence $n=r+1$. 
\end{proof}

 \vspace*{1mm}
\small{
{\bf Acknowledgements.} 
This work was supported by the project MODFLAT of the 
European Research Council and the NCCR SwissMAP of the Swiss National Science Foundation,
 and in part by the Russian Fund for Basic Research Grant No. 18--01--00271.
}

 \vspace*{1.5mm}
{\sc \small
\noindent
Section of Mathematics, University of Geneva,  
C.P. 64, 1211 Gen\`eve 4, Switzerland   \\[1mm]
Steklov Mathematical Institute,
Russian Academy of Sciences,
Fontanka 27, 191023, St. Petersburg, Russia}

\end{document}